%% file: main.tex
\documentclass[11pt]{article}

\usepackage[margin=1in]{geometry}

\usepackage[shortlabels]{enumitem}
\usepackage{array}
\usepackage{tabularx}
\usepackage{wrapfig}
\usepackage{float}
\usepackage{graphicx}
\usepackage{comment}
\usepackage{amsmath,amsthm}
\usepackage{amsfonts}
\usepackage{hyperref}
\usepackage{algorithm}
\usepackage[noend]{algpseudocode}
\usepackage{tikz}
\input{insbox}

\usepackage{enumitem}
\SetEnumitemKey{midsep}{topsep=3pt, itemsep=0pt, partopsep=0pt}
\setlist{midsep}

\usepackage[noframe]{showframe}
\usepackage{framed}
\usepackage{lipsum}
\usetikzlibrary{automata, positioning, calc, fit, shapes.misc}

 {\endMakeFramed}
\definecolor{shadecolor}{gray}{0.90}

\DeclareMathOperator{\rank}{rank}

\newcommand{\YD}{\overline{Y}}

\renewcommand{\P}[1]{\mathbb{P}\left[{#1}\right]}
\newcommand{\E}[1]{\mathbb{E}\left[{#1}\right]}

\makeatletter
\newcommand\Ps@textstyle[2]{\mathbb{P}_{#1}\left[{#2}\right]}
\newcommand\Es@textstyle[2]{\mathbb{E}_{#1}\left[{#2}\right]}

\newcommand\Ps[2]{%
  \mathchoice 
  {\underset{{#1}}{\mathbb{P}}\left[{#2}\right]
  }{\Ps@textstyle{#1}{#2}}{\Ps@textstyle{#1}{#2}}{\Ps@textstyle{#1}{#2}}
}
\newcommand\Es[2]{%
  \mathchoice 
  {\underset{{#1}}{\mathbb{E}}\left[{#2}\right]
  }{\Es@textstyle{#1}{#2}}{\Es@textstyle{#1}{#2}}{\Es@textstyle{#1}{#2}}
}
\makeatother

\newcommand{\D}{\mathcal{D}}
\renewcommand{\H}{\mathcal{H}}

\usepackage{aliascnt}

\hypersetup{colorlinks=true,citecolor=blue,linkcolor=blue}
\usepackage{hyperref}  
\hypersetup{colorlinks=true}
\hypersetup{linkcolor=[rgb]{.7,0,0}}
\hypersetup{citecolor=[rgb]{0,.7,0}}
\hypersetup{urlcolor=[rgb]{.7,0,.7}}
\hypersetup{pdfauthor={Clayton Thomas}}

\newtheorem{theorem}{Theorem}[section]
\newtheorem*{theorem*}{Theorem}

\newaliascnt{definition}{theorem}
\newtheorem{definition}[definition]{Definition}
\aliascntresetthe{definition}

\newtheorem*{definition*}{Definition}

\newaliascnt{lemma}{theorem}
\newtheorem{lemma}[lemma]{Lemma}
\aliascntresetthe{lemma}

\newtheorem*{lemma*}{Lemma}

\newaliascnt{claim}{theorem}
\newtheorem{claim}[claim]{Claim}
\aliascntresetthe{claim}

\newtheorem*{claim*}{Claim}

\newaliascnt{fact}{theorem}

\aliascntresetthe{fact}

\newtheorem*{fact*}{Fact}

\newaliascnt{observation}{theorem}

\aliascntresetthe{observation}

\newtheorem*{observation*}{Observation}

\newaliascnt{conjecture}{theorem}

\aliascntresetthe{conjecture}

\newtheorem*{conjecture*}{Conjecture}

\newaliascnt{corollary}{theorem}

\aliascntresetthe{corollary}

\newtheorem*{corollary*}{Corollary}

\newaliascnt{remark}{theorem}

\aliascntresetthe{remark}

\newtheorem*{remark*}{Remark}

\newaliascnt{proposition}{theorem}

\aliascntresetthe{proposition}

\newtheorem*{proposition*}{Proposition}

\newaliascnt{example}{theorem}

\aliascntresetthe{example}

\newtheorem*{example*}{Example}
\newcommand{\notshow}[1]{}

\makeatletter
\newcommand{\ALC@uniqueautorefname}{Line}
\makeatother

\tikzset{cross/.style={cross out, draw=black, minimum size=2*(#1-\pgflinewidth), inner sep=0pt, outer sep=0pt},
cross/.default={1pt}}
 
\title{The Short-Side Advantage in Random Matching Markets}

\author{
  Linda Cai \\
  Princeton University \\
  tcai@princeton.edu
  \and
  Clayton Thomas \\
  Princeton University \\
  claytont@princeton.edu
}

\begin{document}

  \maketitle

\begin{abstract}
  A breakthrough of Ashlagi, Kanoria, and
  Leshno~\cite{AshlagiUnbalancedCompetition17} found that \emph{imbalance}
  in the number of agents on either side of a random
  matching market has a profound effect on the market's expected
  characteristics.
  Specifically, across all stable matchings,
  the ``long side'' (i.e. the side with a greater number of agents)
  receives significantly worse matches in expectation than the short side.
  Intuitively, this occurs because an agent on the long side is
  \emph{essentially unneeded} to create a stable matching --
  a matching could form almost as easily without them.
  Thus, an agent on the long side has very little market power,
  and must settle for a match which is not much better than a random
  assignment.

  We provide a new and simpler proof for a result
  of~\cite{AshlagiUnbalancedCompetition17} which formalizes this intuition.
\end{abstract}

%
%


\section{Introduction}

\input{Intro.tex}

\section{Preliminaries}
  \label{sectionPrelims}

\input{Prelims.tex}

  \subsection{The Balanced Case} \label{sectionBalanced}

\input{BalancedCase.tex}

\section{Unbalanced Markets and the Effect of Competition}
  \label{sectionUnbalanced}\label{sectionRandom}
  \input{UnbalancedCase.tex}


\bibliographystyle{alpha}
\bibliography{Matching}{}
 
%
%
%

\end{document}

%% file: Intro.tex
Stable matching mechanisms are ubiquitous in theory and in practice,
with centralized mechanisms used to match students to schools or medical
residents to hospitals.
The stability requirement says that no pair of unmatched agents would want
to break the match with their assigned partner and pair with each other
instead.
A vast and classical literature investigates the properties of the set
of stable matchings.
A statistic of primary interest is the \emph{average rank} that agents
receive (i.e. the index each agent has for their partner on their
preference list, where lower is better).
Classical results~\cite{WilsonAnalysisStable72, PittelAverageStable89}
show that in markets with $n$ uniform agents on each side,
the set of stable matchings is large, and average ranks can vary
widely across the set of stable matching.
In the two extreme points of the set (the ``one side optimal'' stable
matchings), one side gets their $O(\log n)$th most preferred partner, and
the other side gets their $\Omega(n / \log n)$th most preferred.


A striking property of these markets is that they are highly dependent on
the \emph{balance} of the number of agents on each side.
Specifically, as shown in~\cite{AshlagiUnbalancedCompetition17},
when there are $n$ agents on the ``short side'' and $n+1$
agents on the ``long side'',
the short side gets their $O(\log n)$th most preferred partner,
and the long side gets their $\Omega(n / \log n)$th most preferred partner,
and this holds in \emph{every} stable matching.

In this paper, we provide a new and intuitive proof for a simplified version
of this foundational result.

\paragraph{Organization}
We recall the definitions, a few basic tools, and the analysis of the
balanced case in \autoref{sectionPrelims}.
We present our main contribution in \autoref{sectionUnbalanced}.



%% file: Prelims.tex
A two sided matching environment consists of a collection 
$\D$ of ``doctors'' and $\H$ of ``hospitals''.
Each doctor $d\in \D$ has a set of preferences over all hospitals $\H$,
represented as (partial) list ordered from most preferred to 
least preferred, and vice-versa.
If a hospital is not on doctor $d$'s list, this means the hospital is an
``unacceptable'' match to $d$ (and vice-versa).
We also denote preferences of $d\in \D$ as a strict order $\succ_d$,
where $h_1 \succ_d h_2$ indicates that doctor $d$ prefers hospital $h_1$
over $h_2$, and vice-versa.
We write $h \succ_d \emptyset$ if $h$ is acceptable to $d$,
and $\emptyset \succ_d h$ otherwise.
We let $P = \big( (\succ_d)_{d\in\D}, (\succ_h)_{h\in\H} \big)$
denote a complete set of preferences for each agent.
Note that each doctor ranks each hospital and vice-versa.
A \emph{matching} is a set of vertex disjoint
edges in the complete bipartite graph on $\D\cup\H$.
We denote a matching as a function 
$\mu : \D\cup \H\to \D\cup \H\cup\{\emptyset\}$,
where $\mu(i)$ is the matched partner of
agent $i$, or $\mu(i) = \emptyset$ if agent $i$ is unmatched.

For a fixed set of preferences, matching $\mu$ is \emph{stable}
if there does not exist a pair $(d,h) \in \D\times\H$ which is not matched
in $\mu$ such that $h \succ_d \mu(d)$ and $d \succ_h \mu(h)$.
A pair $(d,h)$ is said to \emph{block} a matching $\mu$
if $d \succ_h \mu(h)$ and $h \succ_d \mu(d)$.
A \emph{pair} $(d,h)$ is called stable for $P$ if $\mu(d)=h$ in 
\emph{some} stable matching,
and $d$ is called a stable partner of $h$ (and vice-versa).
The canonical method of finding some stable matching is the
one-side-proposing deferred acceptance algorithm.
We denote this as $DPDA$ for the doctor-proposing algorithm,
and $HPDA$ for the hospital-proposing algorithm.




\begin{algorithm}
  \caption{$DPDA(P)$: Doctor-proposing deferred acceptance with preferences $P$}
  \label{algoMPDA}
\begin{algorithmic}[1]
  \State Let $U = \D$ be the set of unmatched doctors
  \State Let $\mu$ be the all empty matching
  \While { $U\ne \emptyset$ and some $d\in U$ has not proposed to every hospital}
    \State Pick any such $d$
    \State \label{line:propose} $d$ ``proposes to'' their favorite hospital $h$ which 
      they have not yet proposed to
      
    \If { $d \succ_h \mu(h)$ }
      \Comment{$h$ ``accepts $d$'s proposal''}
      \label{line:accept}
      \State If $\mu(h)\ne \emptyset$, add $\mu(h)$ to $U$
      \State Set $\mu(h) = d$, remove $d$ from $U$ 
    \EndIf
  \EndWhile
\end{algorithmic}
\end{algorithm}

\begin{theorem}[\cite{GaleS62}]
  \autoref{algoMPDA} always computes a stable matching $\mu_0$.
  Moreover, this is the \emph{doctor-optimal} stable outcome
  (that is, every doctor is matched in $\mu_0$ to their 
  favorite stable partner).
  In particular, the resulting matching is independent of the execution order.
\end{theorem}

We also need the following classical fact:

\begin{theorem}[Rural Hospital Theorem, \cite{RothRuralHospital86}] 
  \label{claimRuralDoctors}
  For any set of preferences,
  the set of unmatched agents is the
  same across every stable outcome.
\end{theorem}

We study uniformly random complete preference lists.
That is, each doctor has one of the $|\H|!$ possible preferences ranking
all of the hospitals, chosen uniformly at random.
In these markets, $DPDA$ can be implemented as a stochastic process using
the ``principle of deferred decisions''. 

\begin{definition}
  \label{def:DPDAPrime}
  Define $DPDA'$ to be the same algorithm as 
  \autoref{algoMPDA}, except 
  preferences are generated as follows:
  \begin{itemize}
    \item On Line~\ref{line:propose}, doctor $d$ chooses a hospital
      uniformly at random from the set of all hospitals.
    \item On Line~\ref{line:accept}, hospital $h$ ignores any proposal from
      a doctor who has already proposed to $h$. If $d$ has not yet proposed
      to $h$ and $h$ has previously seen $k$ proposals from distinct
      doctors, then $h$ accepts the proposal with probability
      $1/(1+k)$.
  \end{itemize}
\end{definition}


\begin{lemma}
  \label{lemmaReproposals}
  The distribution of matches generated by $DPDA'$ is identical to the
  distribution of $DPDA(P)$, where preference $P$
  are generated independently and uniformly randomly for all agents.

  Moreover, the total number of proposals in $DPDA'$ stochastically
  dominates the total number of proposals in $DPDA(P)$.
\end{lemma}
\begin{proof}
  Generating the preference lists of the doctors is equivalent to drawing
  uniformly random hospitals without replacement, and $DPDA'$ ignores 
  repeated proposals.
  If a hospital $h$ sees $k+1$ proposals, then any specific doctor
  who proposed will be $h$'s favorite with probability $1/(1+k)$.
  Moreover, when $h$ sees a new proposal,
  their choice of favorite is independent of the order in which $h$
  saw or accepted proposals in the past.
  Thus, $DPDA'$ is equivalent to $DPDA(P)$ for a uniformly random $P$.
  Finally, observe that we can recover exactly the distribution of
  proposals made in $DPDA(P)$ by filtering out repeated proposals,
  and this coupling proves that
  the number of proposals made in $DPDA(P)$ is stochastically dominated
  by the number of proposals in $DPDA'$.
\end{proof}

By the \emph{rank} a (matched) doctor $d$ gets in a matching $\mu$,
we mean the index of $\mu(d)$ on $d$'s preference list 
(where lower is better).
That is,
\begin{equation*}
  \rank_d(\mu) = 1 + \big| \{ h\in\H : h \succ_d \mu(h) \} \big|
\end{equation*}
In particular, if an agent is unranked, we define their rank as one higher
than their number of acceptable agents.




%% file: BalancedCase.tex
We first consider uniformly random matching markets with $n$ doctors and
$n$ hospitals.
We refer to such a market as \emph{balanced}, in that there are there
are the same number of doctors and hospitals.
It is a classical exercise to show that the expectation
of the average rank that each doctor gets in $DPDA(P)$ is 
$O(\log{n})$.
Correspondingly, the hospitals receive partners which they rank
$\Omega(n/\log n)$ in expectation.
Thus, the proposing side is at a significant advantage.

\begin{theorem}[\cite{WilsonAnalysisStable72}]
  \label{thrm:balancedCase}
  Consider a balanced, uniformly random matching market,
  and let $\mu$ be the result of $DPDA$.
  Then for any $d\in\D$ and $h\in\H$, we have
  \begin{align*}
    \E{\rank_d(\mu)} = O(\log n)
    && \E{\rank_h(\mu)} = \Omega(n/\log n)
  \end{align*}
\end{theorem}
\begin{proof}

  The key insight is that, for random balanced matching markets,
  $DPDA(P)$ essentially behaves like a ``coupon collector'' random variable.
  In more detail:
  \begin{itemize}
    \item The sum of the ranks that the doctors have for their matched
      hospital in $DPDA$ is exactly the total number of proposals made.
    \item For balanced matching markets, $DPDA$ (and $DPDA'$)
      terminates as soon as $n$ distinct hospitals are proposed to.
    \item Let $Y$ denote the total number of proposals in $DPDA'$,
      and let $Z_i$ denote the number of proposals made after the
      $(i-1)$th distinct hospital receives a proposal and before 
      (and including when) the $i$th hospital receives a proposal.
      Thus, $Y = \sum_{i=1}^n Z_i$.
      Because each proposal in $DPDA'$ is uniformly random,
      $Z_i$ is distributed as a geometric random variable with success
      parameter $(n - i+1)/n$.
      Thus,
      \begin{equation*}
        \E{Y} = \sum_{i=1}^n \E{Z_i} = \sum_{i=1}^n \frac{n}{n-i+1}
        = n H_n
      \end{equation*}
      where $H_n = \Theta(\log n)$ denotes the $n$th harmonic number.
  \end{itemize}

  By \autoref{lemmaReproposals}, the total number of proposals in $DPDA(P)$
  is stochastically dominated by $Y$. Moreover, the number of proposals each
  doctor makes (i.e. the rank that doctor receives) is an identically
  distributed random variable, so the expected rank a doctor receives is at
  most $H_n = O(\log n)$.

  On the other hand, the total proposals are equally likely to be received by
  any specific hospital, so the expected number of proposals a hospital
  receives is at most $H_n = O(\log n)$.
  Fix a hospital $h$ and let $Y_h$ denote the number of
  proposals $h$ receives.
  We can calculate the expected rank $h$ receives as
  \begin{align*}
    & \Es{y\sim Y_h}{ 
      \E{ \text{rank of $h$'s favorite doctor out of $y$ proposals} }
    } \\
    & \qquad =
    \Es{y\sim Y_h}{ \frac{n}{1+y} } 
    \ge \frac{n}{1+\E{Y_h}} \ge \Omega\left( \frac{n}{\log n} \right),
  \end{align*}
  where the first inequality follow by Jensen's inequality.
\end{proof}

%


%% file: UnbalancedCase.tex




We now turn to study ``unbalanced'' random matching markets,
i.e. ones in which one side
(say, the hospitals) has strictly more agents than the other side.
For the sake of simplicity, we assume there are $n$
doctors and exactly $n+1$ hospitals.

By adapting the techniques of \autoref{sectionBalanced},
we show that in the hospital-optimal matching,
the expected rank a hospital gets
is $\Omega\big(\frac{n}{\log{n}}\big)$.
Up to constants, this matches the expected rank the hospital would receive in
doctor-optimal stable matching.
Informally, we conclude that in \emph{any} stable matching, the
``short side'' essentially picks their matches and the ``long side''
is essentially picked, regardless of who is actually the proposing side.

\subsection{Intuition and a ``first attempt''}

The stark difference between the balanced and unbalanced case is
counter-intuitive at first glance.
Thus, we first give some intuition for why one might expect
the unbalanced case to differ fundamentally from the balanced case.
Consider \emph{hospital-proposing} deferred acceptance
with $n+1$ hospitals and $n$ doctors.
In the balanced case, this random process would terminate when every doctor
has been proposed to once.
But now, because some hospital must go unmatched, the algorithm will only terminate
once \emph{some hospital has proposed to every doctor}. This is a very different
random process, and one can imagine it must run for a much longer time, forcing
the proposing hospitals to be matched to much worse partners than in the balanced
case. Unfortunately, this random process is fairly difficult to
analyze\footnote{
  For instance, to get a useful analysis, we'd need to keep track of which
  hospital is currently proposing, which doctors they have already proposed to,
  and how likely each doctor is to accept a new proposal.
}, so we take a different approach. 

\subsection{Proof}
\label{sectionUnbalancedProof}

The first key to our proof is the following claim,
which is based on the concept of ``list truncation''
as in~\cite{ImmorlicaHonestyStability05}.

\begin{lemma}[Adapted from~\cite{ImmorlicaHonestyStability05}]
  \label{claimWomenTruncateToGood}
  Fix a hospital $h^*$. For any preference set $P$ and $i\in[n]$,
  let $P_i$ denote altering $P$ by truncating the preference
  list of $h^*$ at place $i$ 
  (i.e. marking doctors ranked worse than $i$ unacceptable).
  Then $h^*$ has a stable partner of rank better than $i$
  if and only if $h^*$ is matched in $DPDA(P_i)$.
\end{lemma}

\begin{proof}
  ($\Rightarrow$)
  Suppose for contradiction that $h^*$ has a stable partner of rank better than $i$,
  yet $h^*$ is not matched in $DPDA(P_i)$.
  Let $\mu = DPDA(P_i)$, and
  let $\mu'$ be any stable matching (for preferences $P$)
  where $h^*$ receives a partner better than rank $i$.
  Observe that $\mu'$ must be stable for preferences $P_i$,
  because any pair blocking for $P_i$ must be blocking for preferences $P$.
  Thus, the set of matched agents in $\mu$ would be identical to that of $\mu'$,
  by \autoref{claimRuralDoctors}.
  In particular, $h^*$ should be matched in $\mu'$, a contradiction.

  ($\Leftarrow$) Now, suppose $h^*$ is matched in $\mu' = MPDA(P_i)$.
  We know $\mu'(h^*)$ is ranked by $h^*$ better than $i$, so it suffices
  to prove that $\mu'$ is stable for preferences $P$.
  Certainly, $\mu'$ is stable for preferences $P_i$.
  Why might a matching, stable for $P_i$, not be stable for preferences $P$?
  The only way is if the blocking pair $(d,h)$ for $P$
  is such that $h=h^*$ and $h^*$ truncated $d$ off their preference list.
  But $h^*$ only accepts proposals from doctor ranked better than $i$,
  and they got a match, so they can't possibly be matched below $d$.
  Thus $\mu'$ is stable for preferences $P$,
  and $h^*$ has a stable partner of rank better than $i$.
\end{proof}

With the above lemma, the proof sketch is as follows:
\begin{itemize}
  \item By \autoref{claimWomenTruncateToGood},
    hospital $h^*$'s rank for their partner in hospital-optimal
    outcome is equal to the minimum
    rank $i$ at which they can truncate their list while still being matched in
    $DPDA(P_i)$.

  \item Consider running $DPDA'$, modified such that 
    $h^*$ {\bf rejects all proposals they receive}.
    Now, $DPDA'$ terminates exactly when all hospitals \emph{other than $h^*$}
    receive a match.
    The number of proposals again follows a coupon-collector-like
    random variable, and we expect $O(n\log n)$ total proposals. 
    In particular, $h^*$ should get $O(\log n)$ total proposals 
    before $DPDA'$ terminates.

  \item Thus, in expectation, the rank of the best proposal $h^*$ received
    is $\Omega(n / \log n)$.
    Thus, in expectation $h^*$ has no stable partners better than this rank.
\end{itemize}

Below we make this intuition formal.

\begin{theorem}\label{thrmUnbalancedExpectation}
  Consider a uniformly random matching market 
  with $n$ doctors and $n+1$ doctors,
  and let $\mu$ denote {\bf \emph{any}} stable matching.
  For any $h\in\H$, we have
  \begin{align*}
    \E{\rank_h(\mu)} = \Omega(n/\log n)
  \end{align*}
\end{theorem}
\begin{proof}
  Fix a hospital $h^*$. For uniformly randomly generate preferences $P$,
  let $P_n$ denote altering $h^*$'s preference list by submitting
  an empty list (i.e. marking all doctors unacceptable).
  Let $Y$ denote the number of proposals $h^*$ receives in $DPDA(P_n)$.

  \begin{claim} \label{Claim:constantBoundOnCC}
    $Y \leq 3\log{n}$ with probability at least $1/2$.
  \end{claim}
  \begin{proof}
    Let $DPDA'_n$ denote running $DPDA'$ as in \autoref{def:DPDAPrime},
    but where $h^*$ rejects every proposal they receive.
    Let $\YD$ denote the number the number of proposals $h^*$
    receives in $DPDA'_n$.
    By \autoref{lemmaReproposals},
    $\YD$ stochastically dominates $Y$.
    For $i\in[n]$,
    let $Z_i$ the number of proposals in $DPDA_n'$ after the $(i-1)$th
    distinct hospital in $\H\setminus\{h^*\}$ receives a proposal,
    and before (and including when) the $i$th distinct hospital in 
    $\H\setminus\{h^*\}$ receives a proposal.
    As each proposal in $DPDA_n'$ is independently distributed across the
    $n+1$ hospitals, and there are $n+1-i$ hospitals in
    $\H\setminus\{h^*\}$ who have not yet been proposed to,
    $Z_i$ is distributed as a
    geometric random variable with success parameter $(n+1-i)/(n+1)$.
    Thus, $\E{Z_i} = (n+1)/(n+1-i)$.

    Let $Y_i$ denote the number of proposals in $Z_i$ which go to $h^*$.
    Conditioned on $Z_i = z$, the first $z-1$ proposals among the proposals
    in $Z_i$ are uniformly distributed among $h^*$ and the $i-1$ hospitals
    in $\H\setminus\{h^*\}$ who have received a proposal.
    Thus, each of these $z-1$ proposals goes to $h^*$ with probability
    $1/i$. The last of the $z$ proposal in $Z_i$ never goes to $h^*$, by
    definition. Thus, $\E{Y_i | Z_i = z} = (z-1)/i$, and we have:
    \begin{align*}
        \E{\YD} &
        = \sum_{i=1}^n \E{Y_i} 
        = \sum_{i=1}^n \E{ \E{Y_i | Z_i} } 
        = \sum_{i=1}^n \E{ (Z_i - 1)/i } 
        \\ &
        = \sum_{i=1}^n \frac 1 i \left(\frac{n+1}{(n+1-i)} - 1 \right)
        = \sum_{i=1}^n \frac 1 i \frac{i}{n+1-i} 
        = \sum_{i=1}^n \frac{1}{n+1-i} 
        = H_n
    \end{align*}

  By Markov's inequality and the fact that $\YD$ stochastically dominates
  $Y$, we get
  \begin{align*}
    \P{Y \le 3\log n} \ge \P{\YD \le 3\log n } 
      \ge \P{\YD \le 2 H_n } \ge 1/2
  \end{align*}

  \end{proof}

  By \autoref{claimWomenTruncateToGood}, the rank $h^*$ gets in the
  hospital-optimal stable matching $\mu$ is exactly the rank of their 
  favorite doctor among the $Y$ who proposed to $h^*$ during $DPDA_n$.
  As the $n$ doctors are uniformly ranked by $h^*$, we have
  \begin{align*}
    \E{ \rank_{h^*}(\mu) }
    & \ge \P{ Y \le 3\log n } \E{ \rank_{h^*}(\mu) | Y \le 3\log n } \\
    & \ge \frac 1 2 \cdot \frac n {3\log n}
      = \Omega\left(\frac n {\log n}\right)
  \end{align*}
  As $\mu$ is the hospital-optimal stable matching,
  this holds for all other stable matchings as well.



\end{proof}
